\documentclass[12pt,letterpaper]{article}
\usepackage{amsmath,amsfonts,amsthm,amssymb}

\theoremstyle{plain}

\numberwithin{equation}{section}
\newtheorem{thm}{Theorem}[section]
\newtheorem{lem}[thm]{Lemma}
\newtheorem{cor}[thm]{Corollary}

\newenvironment{exam}[1]
{\begin{flushleft}\textbf{Example #1}.\enspace}%
{\end{flushleft}}

\newcounter{cond}

\newcommand{\positive}{{\mathbb N}}

\newcommand{\real}{{\mathbb R}}
\newcommand{\ascript}{{\mathcal A}}
\newcommand{\bscript}{{\mathcal B}}
\newcommand{\bhat}{{\widehat{b}}}

\newcommand{\cupdot}{\mathbin{\cup{\hskip-5.4pt}^\centerdot}\,}

\newcommand{\bigcupdotmonen}{{\bigcup _{i=m+1}^n}{\hskip-15pt}^\centerdot{\hskip 15pt}}
\newcommand{\bigcupdotmtwon}{{\bigcup _{i=m+2}^n}{\hskip-15pt}^\centerdot{\hskip 15pt}}
\newcommand{\bigcupdotmthreen}{{\bigcup _{i=m+3}^n}{\hskip-15pt}^\centerdot{\hskip 15pt}}
\newcommand{\bigcupdotim}{{\bigcup _{i=1}^m}{\hskip-8pt}^\centerdot{\hskip 8pt}}
\newcommand{\bigcupdotimminus}{{\bigcup _{i=1}^{m-1}}{\hskip-10pt}^\centerdot{\hskip 10pt}}
\newcommand{\bigcupdotitwo}{{\bigcup _{i=1}^2}{\hskip-8pt}^\centerdot{\hskip 8pt}}
\newcommand{\bigcupdotithree}{{\bigcup _{i=1}^3}{\hskip-8pt}^\centerdot{\hskip 8pt}}

\newcommand{\ab}[1]{\left|#1\right|}
\newcommand{\brac}[1]{\left\{#1\right\}}
\newcommand{\paren}[1]{\left(#1\right)}
\newcommand{\sqbrac}[1]{\left[#1\right]}
\newcommand{\parsq}[1]{\left(#1\right]}

\begin{document}

\title{EXAMPLES OF QUANTUM INTEGRALS}
\author{Stan Gudder\\ Department of Mathematics\\
University of Denver\\ Denver, Colorado 80208\\
sgudder@math.du.edu}
\date{}
\maketitle

\begin{abstract}
We first consider a method of centering and a change of variable formula for a quantum integral. We then present three types of quantum integrals. The first considers the expectation of the number of heads in $n$ flips of a ``quantum coin.'' The next computes quantum integrals for destructive pairs examples. The last computes quantum integrals for a (Lebesgue)${}^2$ quantum measure. For this last type we prove some quantum counterparts of the fundamental theorem of calculus.

\end{abstract}

\section{Introduction}  
Quantum measure theory was introduced by R.~Sorkin in his studies of the histories approach to quantum mechanics and quantum gravity \cite{sor94}. Since then, he and several other authors have continued this study
\cite{gt09, gud1, sal02, sor07, sorapp, sw08} and the author has developed a general quantum measure theory for infinite cardinality spaces \cite{gud2}. Very recently the author has introduced the concept of a quantum integral
\cite{gud3}. Although this integral generalizes the classical Lebesgue integral, it may exhibit unusual behaviors that the Lebesgue integral does not. For example, the quantum integral may be nonlinear and nonmonotone. Because of this possible nonstandard behavior we lack intuition concerning properties of the quantum integral. To help us gain some intuition for this new integral, we present various examples of quantum integrals.

The paper begins with a method of centering and a change of variable formula for a quantum integral. Examples of centering and variable changes are given. We also consider quantum integrals over subsets of the measure spaces. We then present three types of quantum integrals. The first considers the expectation $a_n$ of the number of heads in $n$ flips of a ``quantum coin.'' We prove that $a_n$ asymptotically approaches the classical value $n/2$ as $n$ approaches infinity, and numerical data are given to illustrate this. The next computes quantum integrals for destructive pairs examples. The functions integrated in these examples are monomials. The last computes quantum integrals for (Lebesgue)${}^2$ quantum measure. For this last type , some quantum counterparts of the fundamental theorem of calculus are proved.

\section{Centering and Change of Variables} 
If $(X,\ascript )$ is a measurable space, a map $\mu\colon\ascript\to\real ^+$ is \textit{grade}-2 \textit{additive}
\cite{gt09, gud1, gud2, sor94}, if
\begin{equation*}
\mu\paren{A\cupdot B\cupdot C}=\mu\paren{A\cupdot B}+\mu\paren{A\cupdot C}+\mu\paren{B\cupdot C}
  -\mu (A)-\mu (B)-\mu (C)
\end{equation*}
for any mutually disjoint $A,B,C\in\ascript$ where $A\cupdot B$ denotes $A\cup B$ whenever
$A\cap B\ne\emptyset$. A $q$-\textit{measure} is a grade-2 additive map $\mu\colon\ascript\to\real ^+$ that also satisfies the following continuity conditions \cite{gud2}
\begin{list} {(C\arabic{cond})}{\usecounter{cond}
\setlength{\rightmargin}{\leftmargin}}
\item For any increasing sequence $A_i\in\ascript$ we have
\begin{equation*}
\mu\paren{\cup A_i}=\lim _{i\to\infty}\mu (A_i)
\end{equation*}
\item For any decreasing sequence $B_i\in\ascript$ we have
\begin{equation*}
\mu\paren{\cap A_i}=\lim _{i\to\infty}\mu (A_i)
\end{equation*}
\end{list}

A $q$-measure $\mu$ is not always additive, that is, $\mu\paren{A\cupdot B}\ne\mu (A)+\mu (B)$ in general. A
$q$-\textit{measure space} is a triple $(X, \ascript ,\mu )$ where $(X,\ascript )$ is a measurable space and
$\mu\colon\ascript\to\real ^+$ is a $q$-measure \cite{gud1, gud2, sor94}. Let $(X,\ascript ,\mu )$ be a $q$-measure space and let $f\colon X\to\real$ be a measurable function. It is shown in \cite{gud3} that the following real-valued functions of $\lambda\in\real$ are Lebesgue measurable:
\begin{equation*}
\lambda\mapsto\mu\paren{\brac{x\colon f(x)>\lambda}},\quad\lambda\mapsto\mu\paren{\brac{x\colon f(x)<-\lambda}}
\end{equation*}
We define the \textit{quantum integral} of $f$ to be
\begin{equation}
\label{eq21}              
\int fd\mu =\int _0^\infty\mu\paren{\brac{x\colon f(x)>\lambda}}d\lambda
-\int _0^\infty\mu\paren{\brac{x\colon f(x)<-\lambda}}d\lambda
\end{equation}
where $d\lambda$ is Lebesgue measure on $\real$. If $\mu$ is an ordinary measure (that is; $\mu$ is additive) then
$\int fd\mu$ is the usual Lebesgue integral \cite{gud3}. The quantum integral need not be linear or monotone. That is,
$\int (f+g)d\mu\ne\int fd\mu +\int gd\mu$ and $\int fd\mu\not\le\int gd\mu$ whenever $f\le g$, in general. However, the integral is homogenious in the sense that $\int\alpha fd\mu =\alpha\int fd\mu$, for $\alpha\in\real$.

Definition \eqref{eq21} gives the number zero a special status which is unimportant when $\mu$ is a measure, but which makes a nontrivial difference when $\mu$ is a general $q$-measure. It may be useful in applications to define for $a\in\real$ the $a$-\textit{centered quantum integral}
\begin{align}
\label{eq22}              
\int fd\mu _a&=\int _a^\infty\mu\sqbrac{f^{-1}(\lambda ,\infty )}d\lambda
  -\int _{-a}^\infty\mu\sqbrac{f^{-1}(-\infty, -\lambda )}d\lambda\notag\\
  &=\int _a^\infty\mu\sqbrac{f^{-1}(\lambda ,\infty )}d\lambda
  -\int _{-\infty}^a\mu\sqbrac{f^{-1}(-\infty ,\lambda )}d\lambda
\end{align}
Of course, $\int fd\mu _0=\int fd\mu$ but we shall omit the subscript $0$. Our first result shows how to compute
$\int fd\mu _a$ when $f$ is a simple function.

\begin{lem}       
\label{lem21}
Let $f=\sum _{i=1}^n\alpha _i\chi _{A_i}$ be a simple function with $A_i\cap A_j=\emptyset$, $i\ne j$,
$\cupdot\,_{i=1}^nA_i=X$. If $\alpha _1<\cdots\alpha _m\le a<\alpha _{m+1}<\cdots <\alpha _n$ then
\begin{align}
\label{eq23}              
\int fd\mu _a&=(\alpha _{m+1}-a)\mu\paren{\bigcupdotmonen A_i}+(\alpha _{m+2}-\alpha _{m-1})
  \mu\paren{\bigcupdotmtwon A_i}\notag\\\noalign{\smallskip}
  &\quad +\cdots +(\alpha _n-\alpha _{n-1})\mu (A_n)-\left [(a-\alpha _m)\mu\paren{\bigcupdotim A_i}\right .
  \notag\\\noalign{\smallskip}
  &\quad \left .-(\alpha _m-\alpha _{m-1})\mu\paren{\bigcupdotimminus A_i}
  +\cdots +(\alpha _2-\alpha _1)\mu (A_1)\right ]\displaybreak\\
\label{eq24}              
  &=\alpha _1\mu (A_1)+\alpha _2[\mu\paren{A_1\cupdot A_2}-\mu (A_1)]+\cdots +
  \alpha _m[\mu\paren{A_1\cupdot A_m}\notag\\
  &\quad +\cdots +\mu\paren{A_{m-1}\cupdot A_m}-\mu (A_1)-\cdots -\mu (A_{m-1})-(m\!-\!2)\mu (A_m)]\notag\\
  &\quad +\alpha _{m+1}[\mu\paren{A_{m+1}\cupdot A_{m+2}}+\cdots +\mu\paren{A_{m+1}\cupdot A_n}\notag\\
  &\quad -(n-m-2)\mu (A_{m+1})-\mu (A_{m-2})-\cdots -\mu (A_n)]\notag\\
  &\quad +\cdots +\alpha _{n-1}[\mu\paren{A_{n-1}\cupdot A_n}-\mu (A_n)]+\alpha _n\mu (A_n)\notag\\
  \noalign{\smallskip}
  &\quad -a\sqbrac{\mu\paren{\bigcupdotim A_i}+\mu\paren{\bigcupdotmonen A_i}}
\end{align}
\end{lem}
\begin{proof}
Equation \eqref{eq23} is a straightforward application of the definition \eqref{eq22}. We can rewrite \eqref{eq23} as
\begin{align*}
\int fd\mu _a&=\alpha _{m+1}\sqbrac{\mu\paren{\bigcupdotmonen A_i}-\mu\paren{\bigcupdotmtwon A_i}}\\
\noalign{\smallskip}
  &\quad +\alpha _{m+2}\sqbrac{\mu\paren{\bigcupdotmtwon A_i}-\mu\paren{\bigcupdotmthreen A_i}}\\
  &\quad +\cdots +\alpha _{n-1}\sqbrac{\mu\paren{A_{n-1}\cupdot A_n}-\mu (A_n)}
  +\alpha _n\mu (A_n)+\alpha _1\mu (A_1)\\\noalign{\smallskip}
  &\quad +\alpha _2\sqbrac{\mu\paren{\bigcupdotitwo A_i}-\mu (A_1}
  +\alpha _3\sqbrac{\mu\paren{\bigcupdotithree A_i}-\mu\paren{\bigcupdotitwo A_i}}\\\noalign{\smallskip}
  &\quad +\cdots +\alpha _m\sqbrac{\mu\paren{\bigcupdotim A_i}-\mu\paren{\bigcupdotimminus A_i}}\\
  \noalign{\smallskip}
  &\quad -a\sqbrac{\mu\paren{\bigcupdotim A_i}+\mu\paren{\bigcupdotmonen A_i}}
 \end{align*}
Applying Theorem~3.3 \cite{gud2} the result follows.
\end{proof}

\begin{cor}       
\label{cor22}
If $\mu$ is a measure and $f$ is integrable, then
\begin{equation*}
\int fd\mu _a=\int fd\mu -a\mu (X)
\end{equation*}
\end{cor}
\begin{proof}
By \eqref{eq24} the formula holds for simple functions. Approximate $f$ by an increasing sequence of simple functions and apply the monotone convergence theorem.
\end{proof}

As an illustration of Lemma~\ref{lem21}, let $f=a\chi _A+b\chi _B$ be a simple function with $A\cap B=\emptyset$,
$A\cup B=X$, $0\le a<b$. By \eqref{eq24} we have
\begin{equation*}
\int fd\mu =a\sqbrac{\mu\paren{A\cupdot B}-\mu (B)}+b\mu (B)
\end{equation*}
This also shows that the quantum integral is nonlinear because if $\mu\paren{A\cupdot B}\ne \mu (A)+\mu (B)$ then
\begin{equation*}
\int (a\chi _A+b\chi _B)d\mu =\int fd\mu\ne a\mu (A)+b\mu (B)=a\int\chi _Ad\mu +b\int\chi _Bd\mu
\end{equation*}

Corollary~\ref{cor22} shows that if $\mu$ is a measure, then $\int fd\mu _a$ is just a translation of $\int fd\mu$ by the constant $a\mu (X)$ for all integrable $f$. We now show that this does not hold when $\mu$ is a general
$q$-measure.

\begin{exam}{1}    
Let $a>0$ be a fixed constant and let $f=c\chi _A$ be a simple function with $c\ne 0$ and $A\ne\emptyset , X$. We can write $f$ in the canonical form
\begin{equation*}
f=0\chi _{A'}+c\chi _A
\end{equation*}
where $A'$ is the complement of $A$. By Lemma~\ref{lem21} we have that $\int fd\mu =c\mu (A)$ and applying
\eqref{eq24} to the various cases we obtain the following:\newline
\textbf{Case 1.}\enspace If $0<a<c$, then $\alpha _1=0$, $\alpha _2=c$, $\alpha _1<a<\alpha _2$, $A_1=A'$ and $A_2=A$. We compute
\begin{align*}
\int fd\mu _a&=a\mu (A')+c\mu (A)-a\sqbrac{\mu (A')+\mu (A)}\\
  &=\int fd\mu -a\sqbrac{\mu (A')+\mu (A)}
 \end{align*}
\noindent\textbf{Case 2.}\enspace If $0<c<a$, then $\alpha _1=0$, $\alpha _2=c$, $\alpha _1<\alpha _2<a$, $A_1=A'$ and $A_2=A$. We compute
\begin{align*}
\int fd\mu _a&=0\mu (A')+c\sqbrac{\mu (X)-\mu (A')}-a\mu (X)\\
  &=\int fd\mu -c\sqbrac{\mu (A)+\mu (A')-\mu (X)}-a\mu (X)
\end{align*}
\noindent\textbf{Case 3.}\enspace If $c<0<a$, then $\alpha _1=c$, $\alpha _2=0$, $\alpha _1<\alpha _2<a$,
$A_1=A$ and $A_2=A'$. We compute
\begin{equation*}
\int fd\mu _a=c\mu (A)+0\sqbrac{\mu (X)-\mu (A)}-a\mu (X)=\int fd\mu -a\mu (X)
\end{equation*}
\end{exam}

We now derive a change of variable formula. Suppose $g$ is an increasing and differentiable function on $\real$ and let $g^{-1}(\pm\infty )=\lim\limits _{\lambda\to\pm\infty}g^{-1}(\lambda )$. If $f\colon X\to\real$ is measurable, then so is $g\circ f$ and we have
\begin{align*}
\int g\circ fd\mu _a&=\int _a^\infty\mu\sqbrac{\brac{x\colon g\circ f(x)>\lambda}}d\lambda
  -\int _{-\infty}^a\mu\sqbrac{\brac{x\colon g\circ f(x)<\lambda}}d\lambda\\
  &=\int _a^\infty\mu\!\!\sqbrac{\brac{x\colon f(x)>g^{-1}(\lambda )}}d\lambda
  -\!\int _{-\infty}^a\!\sqbrac{\brac{x\colon f(x)<g^{-1}(\lambda )}}d\lambda
\end{align*}
Letting $t=g^{-1}(\lambda )$, $g(t)=\lambda$, $g'(t)dt=d\lambda$, by the usual change of variable formula we obtain
\begin{align}
\label{eq25}              
\int&g\circ fd\mu _a\\
&=\int _{g^{-1}(a)}^{g^{-1}(\infty )}\mu\sqbrac{\brac{x\colon f(x)>t}}g'(t)dt
  -\int _{g^{-1}(-\infty )}^{g^{-1}(a)}\mu\sqbrac{\brac{x\colon f(x)<t}}g'(t)dt\notag
\end{align}
For example, if $f\ge 0$, letting $g(t)=t^n$ we have
\begin{equation}
\label{eq26}              
\int f^nd\mu=\int _0^\infty\mu\sqbrac{\brac{x\colon f(x)>t}}nt^{n-1}dt
\end{equation}

As with the Lebesgue integral, if $A\in\ascript$ we define
\begin{equation*}
\int _Afd\mu =\int\chi _Ad\mu
\end{equation*}
We then have
\begin{align*}
\int _Afd\mu&=\int _0^\infty\mu\sqbrac{\brac{x\colon\chi _A(x)f(x)>\lambda}}d\lambda
  -\int _0^\infty\mu\sqbrac{\brac{x\colon\chi _A(x)f(x)<-\lambda}}d\lambda\\
  &=\int _0^\infty\mu\sqbrac{A\cap f^{-1}(\lambda ,\infty )}d\lambda
  -\int _0^\infty\mu\sqbrac{A\cap f^{-1}(-\infty ,-\lambda )}d\lambda\\
  &=\int _0^\infty\brac{\mu\sqbrac{A\cap f^{-1}(\lambda ,\infty )}-\mu\sqbrac{A\cap f^{-1}(-\infty ,-\lambda )}}d\lambda
\end{align*}
Now $(A,A\cap\ascript )$ is a measurable space and it is easy to check that $\mu _A(B)=\mu (A\cap B)$ is a
$q$-measure on $A\cap\ascript$ so $(A,A\cap\ascript ,\mu _A)$ is a $q$-measure space. Hence, for a measurable function $f\colon X\to\real$, the restriction $f\mid A\colon A\to\real$ is measurable and
\begin{equation*}
\int _Afd\mu =\int f\mid Ad\mu _A
\end{equation*}
Similar definitions apply to the centered integrals $\int _Afd\mu _a$.

\section{A Quantum Coin} 
If we flip a coin $n$ times the resulting sample space $X_n$ consists of $2^n$ outcomes each being a sequence of
$n$ heads or tails. For example, a possible outcome for 3 flips is HHT and $X_2=\brac{HH,HT,TH,TT}$. If this were an ordinary fair coin then the probability of a subset $A\subseteq X_n$ would be $\ab{A}/2^n$ where $\ab{A}$ is the cardinality of $A$. However, suppose we are flipping a ``quantum coin'' for which the probability is replaced by the
$q$-measure $\mu _n(A)=\ab{A}^2/2^{2n}$. It is easy to check that $\mu$ is indeed a $q$-measure. In fact the square of any measure is a $q$-measure.

Let $f_n\colon X_n\to\real$ be the random variable that gives the number of heads in $n$ flips. For example, $f_3(HHT)=2$. For an ordinary coin the expectation of $f_n$ is $n/2$. We are interested in computing the
``quantum expectation'' $\int f_nd\mu _n$ for a ``quantum coin.'' For the case $n=1$ we have $X_1=\brac{x_1,x_2}$ with $f_1(x_1)=1$, $f_1(x_2)=0$. Then $f_1=\chi _{\brac{x_1}}$ and by \eqref{eq23} we have
\begin{equation*}
\int f_1d\mu _1=\mu _1\paren{\brac{x_1}}=\frac{1}{4}
\end{equation*}
For the case $n=2$, we have $X_2=\brac{x_1,x_2,x_3,x_4}$ with $f_2(x_1)=2$, $f_2(x_2)=f_2(x_3)=1$, $f_2(x_4)=0$. Then
\begin{equation*}
f_2=\chi _{\brac{x_2,x_3}}+2\chi _{\brac{x_1}}
\end{equation*}
and by \eqref{eq23} we have
\begin{equation*}
\int f_2d\mu _2=\mu _2\paren{\brac{x_1,x_2,x_3}}+\mu _2\paren{\brac{x_1}}=\frac{9}{16}+\frac{1}{16}=\frac{5}{8}
\end{equation*}
Continuing this process, $X_3=\brac{x_1,\ldots ,x_8}$ and
\begin{equation*}
f_2=\chi _{\brac{x_5,x_6,x_7}}+2\chi _{\brac{x_2,x_3,x_4}}+3\chi _{\brac{x_1}}
\end{equation*}
By \eqref{eq23} we obtain
\begin{align*}
\int f_3d\mu _3&=\mu _3\paren{\brac{x_1,\ldots ,x_7}}+\mu _3\paren{\brac{x_1,\ldots ,x_4}}
  +\mu _3\paren{\brac{x_1}}\\
  &=\frac{49}{64}+\frac{16}{64}+\frac{1}{64}=\frac{33}{32}
\end{align*}
For 4 flips, $X_4=\brac{x_1,\ldots ,x_{16}}$ and
\begin{align*}
\int f_4d\mu _4&=\mu _4\paren{\brac{x_1,\ldots ,x_{15}}}+\mu _4\paren{\brac{x_1,\ldots ,x_{11}}}\\
  &\quad +\mu _4\paren{\brac{x_1,\ldots ,x_5}}+\mu\paren{\brac{x_1}}\\\noalign{\smallskip}
  &=\frac{15^2+11^2+5^2+1}{16^2}=\frac{93}{64}
\end{align*}
For 5 flips, $X_5=\brac{x_1,\ldots ,x_{32}}$ and
\begin{align*}
\int f_5d\mu _5&=\mu _5\paren{\brac{x_1,\ldots ,x_{31}}}+\mu _5\paren{\brac{x_1,\ldots ,x_{26}}}
  +\mu _5\paren{\brac{x_1,\ldots ,x_{16}}}\\
  &\quad +\mu _5\paren{\brac{x_1,\ldots ,x_6}}+\mu\paren{\brac{x_1}}\\\noalign{\smallskip}
  &=\frac{31^2+26^2+16^2+6^2+1}{32^2}=\frac{965}{512}
\end{align*}

Letting $a_n=\int f_nd\mu _n$ we have that
\begin{align*}
a_1&=\frac{1}{2^2}\ \binom{1}{0}^2\\\noalign{\smallskip}
a_2&=\frac{1}{2^4}\brac{\binom{2}{0}^2+\sqbrac{\binom{2}{0}+\binom{2}{1}}^2}\\\noalign{\smallskip}
a_3&=\frac{1}{2^6}\brac{\binom{3}{0}^2+\sqbrac{\binom{3}{0}+\binom{3}{1}}^2
  +\sqbrac{\binom{3}{0}+\binom{3}{1}+\binom{3}{2}}^2}\\
&\ \vdots\\
a_n&=\frac{1}{2^{2n}}\brac{\binom{n}{0}^2\!\!+\!\!\sqbrac{\binom{n}{0}+\binom{n}{1}}^2\!\!
  +\cdots +\sqbrac{\binom{n}{0}+\binom{n}{1}\!+\!\cdots +\binom{n}{n-1}}^2}
\end{align*}
We shall show that $a_n$ asymptotically approaches the classical value $n/2$ for large $n$; that is
\begin{equation}
\label{eq31}                 
\lim _{n\to\infty}\frac{2a_n}{n}=1
\end{equation}
As numerical evidence for this result the first seven values of $2a_n/n$ are:
0.5000, 0.6250, 0.6875, 0.7266, 0.7539, 0.7749, 0.7905 and the twentieth value is 0.8737. The next result shows that the quantum expectation does not exceed the classical expectation.

\begin{lem}       
\label{lem31}
For all $n\in\positive$, $\int f_nd\mu _n\le n/2$.
\end{lem}
\begin{proof}
Letting $A_i=f_n^{-1}\paren{\brac{i}}$, $i=1,\ldots ,n$, applying \eqref{eq24} we obtain
\begin{align*}
\int f_nd\mu _n&=[\mu _n\paren{A_1\cupdot A_2}+\cdots +\mu _n\paren{A_1\cupdot A_n}
  -(n-2)\mu _n(A_1)\\
  &\quad -\mu _n(A_2)-\cdots -\mu _n(A_n)]\\
  &\quad +2[\mu _n\paren{A_1\cupdot A_3}+\cdots +\mu _n\paren{A_2\cupdot A_n}-(n-3)\mu _n(A_2)\\
  &\quad -\mu _n(A_3)-\cdots -\mu _n(A_n)]\\
  &\quad\ \vdots\\
  &\quad +(n-1)[\mu _n\paren{A_{n-1}\cupdot A_n}-\mu _n(A_n)]+n\mu _n(A_n)\\
  &=\frac{1}{2^{2n}}\left\{\ab{A_1\cupdot A_2}^2+\cdots +\ab{A_1\cupdot A_n}^2-(n-2)\ab{A_1}^2-\ab{A_2}^2\right .\\
  &\quad -\cdots -\ab{A_n}^2\\
  &\quad +2[\ab{A_1\cupdot A_3}^2+\cdots +\ab{A_1\cupdot A_n}^2-(n-3)\ab{A_2}^2-\ab{A_3}^2\\
  &\quad -\cdots -\ab{A_n}^2]\\
   &\quad\ \vdots\\
   &\quad \left .+(n-1)[\ab{A_{n-1}\cupdot A_n}^2-\ab{A_n}^2]+n\ab{A_n}^2\right\}\\
   &=\frac{1}{2^{2n}}\left\{\ab{A_1}^2+2\ab{A_1}\paren{\ab{A_2}+\cdots +\ab{A_n}}\right .\\
   &\quad +2\sqbrac{\ab{A_2}^2+2\ab{A_2}\paren{\ab{A_3}+\cdots +\ab{A_n}}}\\
   &\quad \left .+\cdots +(n-1)\sqbrac{\ab{A_{n-1}}^2+2\ab{A_{n-1}}\ab{A_n}}+n\ab{A_n}^2\right\}
  \end{align*}
By the binomial theorem we conclude that
\begin{align*}
\int f_nd\mu _n&=\frac{1}{2^{2n}}\left\{\ab{A_1}\sqbrac{1+2\paren{2^n-\ab{A_1}-1}}\right .\\
  &\quad +2\ab{A_2}\sqbrac{1+2\paren{2^n-\ab{A_2}-\ab{A_1}-1}}\\
  &\quad +\cdots +(n-1)\ab{A_{n-1}}\sqbrac{1+2\paren{2^n-\ab{A_{n-1}}-\cdots -\ab{A_1}-1}}\\
  &\quad \left .+n\ab{A_n}\sqbrac{2^n-1-\ab{A_1}-\ab{A_2}-\cdots -\ab{A_{n-1}}}\right\}\\
  &\le\frac{1}{2^{2n}}\sqbrac{\ab{A_1}2^n+2\ab{A_2}2^n+\cdots +n\ab{A_n}2^n}\\
  &=\frac{1}{2^{n}}\paren{\ab{A_1}+2\ab{A_2}+\cdots +n\ab{A_n}}=\frac{n}{2}
\end{align*}
where the last equality follows from the classical expectation.
\end{proof}

We now give $a_n$ in closed form and prove \eqref{eq31}.
\begin{thm}       
\label{thm32}
{\rm (a)}\enspace For all $n\in\positive$ we have
\begin{equation*}
a_n=\frac{1}{2}\sqbrac{n+2-\paren{\frac{n\binom{2n}{n}+2}{2^{2n}}}}
\end{equation*}
{\rm (b)}\enspace Equation~\eqref{eq31} holds.
\end{thm}
\begin{proof}
(a)\enspace Letting
\begin{equation*}
b_n=\binom{n}{0}^2+\sqbrac{\binom{n}{0}+\binom{n}{1}}^2+\cdots +\sqbrac{\binom{n}{0}+\binom{n}{1}
  +\cdots +\binom{n}{n-1}}^2
\end{equation*}
it is shown in \cite{hir96} that
\begin{equation*}
b_n=(n+2)2^{2n-1}-\frac{1}{2}n\binom{2n}{n}-1
\end{equation*}
Since $a_n=b_n/2^{2n}$, the result follows.

(b)\enspace By Stirling's formula we have the asymptotic approximation
\begin{equation*}
n!\approx\sqrt{2\pi n\,}\,\frac{n^n}{e^n}
\end{equation*}
for large $n$. Hence,
\begin{align*}
\lim _{n\to\infty}\frac{1}{2^{2n}}\binom{2n}{n}&=\lim _{n\to\infty}\frac{1}{2^{2n}}\frac{(2n)!}{(n!)^2}
  =\lim _{n\to\infty}\frac{1}{2^{2n}}\frac{\sqrt{2\pi n\,}(2n)^{2n}}{e^{2n}}\cdot\frac{e^{2n}}{2\pi nn^{2n}}\\
  &=\lim_{n\to\infty}\frac{1}{\sqrt{2\pi n\,}}=0
\end{align*}
Hence,
\begin{equation*}
\lim _{n\to\infty}\frac{2a_n}{n}=\lim _{n\to\infty}\paren{\frac{n+2}{n}}
  -\lim _{n\to\infty}\frac{\binom{2n}{n}+2}{2^{2n}}=1\qedhere
\end{equation*}
\end{proof}

The next example illustrates the $a$-centered integral $\int f_2d\mu _{2a}$ for two flips of a ``quantum coin.''
\begin{exam}{2}    
The following computations result from applying \eqref{eq23}. If $a\le 0$, then
\begin{align*}
\int f_2d\mu _{2a}&=(0-a)\mu\paren{\brac{x_1,x_2,x_3,x_4}}+\mu\paren{\brac{x_1,x_2,x_3}}
  -\mu\paren{\brac{x_1}}\\
  =\frac{5}{8}-a
\end{align*}
If $0\le a\le 1$, then
\begin{align*}
\int f_2d\mu _{2a}&=(1-a)\mu\paren{\brac{x_1,x_2,x_3}}+\mu\paren{\brac{x_1}}-(a-0)\mu\paren{\brac{x_4}}\\
  &=(1-a)\frac{9}{16}+\frac{1}{16}-\frac{1}{16}a=\frac{5}{8}-\frac{5}{8}a
\end{align*}
If $1\le a\le 2$, then
\begin{align*}
\int f_2d\mu _{2a}&=(2-a)\mu\paren{\brac{x_1}}-(a-1)\mu\paren{\brac{x_2,x_3,x_4}}-\mu\paren{\brac{x_4}}\\
  &=(2-a)\frac{1}{16}-(a-1)\frac{9}{16}-\frac{1}{16}=\frac{5}{8}-\frac{5}{8}a
\end{align*}
If $2\le a$, then
\begin{align*}
\int f_2d\mu _{2a}&=-\sqbrac{(a-2)\mu\paren{\brac{x_1,x_2,x_3,x_4}}+\mu\paren{\brac{x_2,x_3,x_4}}
  +\mu\paren{\brac{x_4}}}\\
  &=-\sqbrac{(a-2)+\frac{9}{16}+\frac{1}{16}}=\frac{11}{8}-a
\end{align*}
We conclude that $\int f_2d\mu _{2n}$ is piecewise linear as a function of $a$.
\end{exam}

\section{Destructive Pairs Examples} 
Consider $X=\sqbrac{0,1}$ as consisting of particles for which pairs of the form $(x,x+3/4)$,
$x\in\sqbrac{0,1/4}$ are \textit{destructive pairs} (or \textit{particle-antiparticle pairs}). Thus, particles in
$x\in\sqbrac{0,1/4}$ annihilate their counterparts in $\sqbrac{3/4,1}$ while particles in $(1/4,3/4)$ do not interact with other particles. Let $\bscript (X)$ be the set of Borel subsets of $X$ and let $\nu$ be Lebesgue measure on
$\bscript (X)$. For $A\in\bscript (X)$ define
\begin{equation*}
\mu (A)=\nu (A)-2\nu\paren{\brac{x\in A\colon x+3/4\in A}}
\end{equation*}
Thus, $\mu (A)$ is the Lebesgue measure of $A$ after the destructive pairs in $A$ annihilate each other. For example, $\mu\paren{\sqbrac{0,1}}=1/2$ and $\mu\paren{\sqbrac{0,3/4}}=3/4$. It can be shown that
$\paren{X,\bscript (X),\mu}$ is a $q$-measure space \cite{gud2}.

Letting $f(x)=x$ and $0<b\le 1$ we shall compute
\begin{equation*}
\int _0^bf(x)d\mu =\int _{\sqbrac{0,b}}f(x)d\mu
\end{equation*}
We first define
\begin{equation*}
F(\lambda )=\mu\paren{\brac{x\colon f\chi _{\sqbrac{0,b}}(x)>\lambda}}
  =\mu\paren{\brac{x\in\sqbrac{0,b}\colon x>\lambda}}=\mu\paren{\parsq{\lambda ,b}}
\end{equation*}
If $b\le 3/4$ then
\begin{equation*}
F(\lambda )=\begin{cases}b-\lambda&\text{if }\lambda\le b\\0&\text{if }\lambda >b\end{cases}
\end{equation*}
We obtain
\begin{equation*}
\int _0^bxd\mu =\int _0^bF(\lambda )d\lambda =\int _0^b(b-\lambda )d\lambda =\frac{b^2}{2}
\end{equation*}
which is the expected classical result because there is no interference (annihilation).

Now suppose that $b\ge 3/4$ in which case there is interference. If $\lambda\ge b-3/4$, then
$F(\lambda )=b-\lambda$ as before. If $\lambda <\,b-3/4$, then
\begin{equation*}
F(\lambda )=(b-\lambda )-2\paren{b-\frac{3}{4}-\lambda}=\frac{3}{2}+\lambda -b
\end{equation*}
We then obtain
\begin{align*}
\int _0^bxd\mu&=\int _0^bF(\lambda )d\lambda =\int _0^{b-3/4}\paren{\frac{3}{2}+\lambda -b}d\lambda
  +\int _{b-3/4}^b(b-\lambda )d\lambda\\\noalign{\medskip}
  &=\frac{3}{2}\,b-\frac{9}{16}-\frac{1}{2}\,b^2
\end{align*}
For example, $\int _0^1xd\mu =7/16$ which is slightly less that $\int _0^1xd\lambda =1/2$. Of course, interference is the cause of this difference. Also, $\int _0^{3/4}xd\mu =9/32$ which agrees with $\int _0^{3/4}xdx$ as shown in the
$b\le 3/4$ case.

We next compute $\int _0^bx^nd\mu$. If $b\le 3/4$, then by our change of variable formula we have
\begin{equation*}
\int _0^bx^nd\mu =n\int _0^b(b-\lambda )\lambda ^{n-1}d\lambda =\frac{b^{n+1}}{n+1}
\end{equation*}
in agreement with the classical result. If $b\ge 3/4$, we obtain by the change of variable formula
\begin{align*}
\int _0^bx^nd\mu&=n\int _0^bF(\lambda )\lambda ^{n-1}d\lambda\\\noalign{\smallskip}
  &=n\sqbrac{\int _0^{b-3/4}\paren{\tfrac{3}{2}+\lambda -b}\lambda ^{n-1}d\lambda
  +\int _{b-3/4}^b(b-\lambda )\lambda ^{n-1}d\lambda}\\\noalign{\smallskip}
  &=\frac{1}{n+1}\sqbrac{b^{n+1}-2\paren{b-\tfrac{3}{4}}^{n+1}}
\end{align*}
As a check, if $n=1$ we obtain our previous result. Notice that the deviation from the classical integral becomes
\begin{equation*}
\int _0^bx^ndx-\int _0^bx^nd\mu =\frac{2}{n+1}\paren{b-\tfrac{3}{4}}^{n+1}
\end{equation*}
which increases as $b$ approaches $1$.

We now change the previous example so that we only have destructive pairs in which case we obtain more interference. We again let $X=\sqbrac{0,1}$, but now we define the $q$-measure
\begin{equation*}
\mu (A)=\nu (A)-2\nu\paren{\brac{x\in A\colon x+\tfrac{1}{2}\in A}}
\end{equation*}
In this case, $(x,x+1/2)$, $x\in\sqbrac{0,1/2}$ are destructive pairs. For instance, $\mu (X)=0$, $\mu\paren{\sqbrac{1/16,5/6}}=1/3$, and $\mu\paren{\sqbrac{0,1/2}}=1/2$. Letting $f(x)=x$ and $0\le a<b\le 1$, we shall compute
\begin{equation*}
\int _a^bxd\mu =\int _{(a,b)}xd\mu
\end{equation*}
We then have
\begin{align*}
F(\lambda )&=\mu\paren{\brac{x\colon f\chi _{(a,b)}(x)>\lambda}}=\mu\paren{\brac{x\in (a,b)\colon x>\lambda}}\\
  \noalign{\smallskip}
  &=\begin{cases}\mu\paren{(a,b)}&\text{if }\lambda\le a\\
  \mu\paren{(\lambda ,b)}&\text{if }a\le\lambda\le b\\0&\text{if }\lambda\ge b\end{cases}
\end{align*}
Now $\brac{x\in (a,b)\colon x+\tfrac{1}{2}\in (a,b)}=\emptyset$ if and only if $b-a\le 1/2$. If $b-a\le 1/2$ we have
\begin{equation*}
F(\lambda )=
\begin{cases}b-a&\text{if }\lambda\le a\\b-\lambda&\text{if }a\le\lambda\le b\\0&\text{if }\lambda\ge b
\end{cases}
\end{equation*}
We then obtain
\begin{equation*}
\int _a^bxd\mu =\int _0^a(b-a)d\lambda +\int _a^b(b-\lambda )d\lambda =\frac{b^2}{2}-\frac{a^2}{2}
\end{equation*}
which is expected because there is not interference.

If $b-a\ge1/2$, letting $c=b-1/2$ we have that $c\ge a$ and
\begin{equation*}
\mu\paren{(a,b)}=b-a-2(c-a)=b-a-2\paren{b-\tfrac{1}{2}-a}=a-b+1
\end{equation*}
If $\lambda\le b-1/2$, then $\mu\paren{(\lambda ,b)}=\lambda -b+1$ and if $\lambda\ge b-1/2$, then
$\mu\paren{(\lambda ,b)}=b-\lambda$. Hence,
\begin{equation*}
F(\lambda )=
\begin{cases}
a-b+1&\text{if }\lambda\le a\\\lambda -b+1&\text{if }a\le\lambda\le b-1/2\\
b-\lambda&\text{if }b-1/2\le\lambda\le b
\end{cases}
\end{equation*}
We conclude that
\begin{align*}
\int _a^bxd\mu&=\int _0^a(a-b+1)d\lambda +\int _a^{b-1/2}(\lambda -b+1)d\lambda
  +\int _{b-1/2}^b(b-\lambda )d\lambda\\
  &=\frac{a^2}{2}-\frac{b^2}{2}+b-\frac{1}{4}
\end{align*}
The deviation from the classical integral becomes
\begin{equation*}
\Delta=\int _a^bxdx-\int _a^bxd\mu =b^2-a^2-b+\frac{1}{4}
\end{equation*}
Notice that $\Delta =0$ if and only if $b=a+1/2$. Special cases of the integral are
\begin{align*}
\int _0^bxd\mu&=-\frac{b^2}{2}+b-\frac{1}{4}\\\noalign{\medskip}
  \int _0^{1/2}xd\mu&=\frac{1}{8}\\\noalign{\medskip}
  \int _0^{3/4}xd\mu&=\frac{7}{32}\\\noalign{\medskip}
  \int _0^{1}xd\mu&=\frac{1}{4}
\end{align*}

\section{(Lebesgue)${}^2$ Quantum Measure} 
We again let $X=\sqbrac{0,1}$ and let $\nu$ be Lebesgue measure on $\bscript (X)$. We define
(Lebesgue)${}^2$ $q$-measure by $\mu (A)=\nu (A)^2$ for $A\in\bscript (X)$ and consider the $q$-measure space
$\paren{X,\bscript (X),\mu}$. The first example in this section is the $a$-centered quantum integral $\int x^nd\mu _a$. Applying the change of variable formula we obtain
\begin{align*}
\int x^nd\mu _a
  &=n\int _a^1\mu\paren{\brac{x\colon x>t}}t^{n-1}dt-n\int _0^a\mu\paren{\brac{x\colon x<t}}t^{n-1}dt\\
  &=n\int _a^1(1-t)^2t^{n-1}dt-n\int _0^at^2t^{n-1}dt\\
  &=\frac{2}{(n+1)(n+2)}-a^n\paren{1-\frac{2n}{n+1}\,a+\frac{2n}{n+2}\,a^2}
\end{align*}
As special cases we have
\begin{align*}
\int xd\mu _a&=\tfrac{1}{3}-a+a^2-\tfrac{2}{3}\,a^3\\
\int x^nd\mu&=\frac{2}{(n+1)(n+2)}
\end{align*}

We now compute the quantum integral $\int _a^bx^nd\mu$ for $0\le a<b\le 1$. Again the change of variable formula gives
\begin{align*}
\int _a^bx^nd\mu&=\int _0^\infty\mu\paren{(a,b)\cap\brac{x\colon x>\lambda ^{1/n}}}d\lambda\\
  &=n\int _0^\infty\mu\paren{(a,b)\cap\brac{x\colon x>t}}t^{n-1}d\lambda\\
  &=n\int _a^b(b-t)^2t^{n-1}dt+n\int _0^a(b-a)^2t^{n-1}dt\\
  &=\frac{2}{(n+1)(n+2)}\,(b^{n+2}-a^{n+2})-\frac{2a^{n+1}}{n+1}\,(b-a)
\end{align*}
As special cases we have
\begin{align*}
\int _a^bxd\mu&=\frac{b^3}{3}-\frac{a^3}{3}-a^2(b-a)\\
\int _a^bd\mu&=(b-a)^2
\end{align*}

We can compute $\int _a^bx^nd\mu$ another way without relying on a change of variables:
\begin{align*}
\int _a^bx^nd\mu&=\int _0^\infty\mu\paren{(a,b)\cap\brac{x\colon x>\lambda ^{1/n}}}d\lambda\\
  &=\int _{a^n}^{b^n}(b-\lambda ^{1/n})^2d\lambda +\int _0^{a^n}(b-a)^2d\lambda\\
  &=\int _{a^n}^{b^n}(b^2-2b\lambda ^{1/n}+\lambda ^{2/n})d\lambda +(b-a)^2a^n\\
  &=\frac{2}{(n+1)(n+2)}\,b^{n+2}-2a^{n+1}\paren{\frac{b}{n+1}-\frac{a}{n+2}}
\end{align*}
which agrees with our previous result.

Until now we have only integrated monomials. We now integrate the more complex function $e^x$. By the change of variable formula
\begin{align*}
\int _a^be^xd\mu&=\int _{-\infty}^\infty\mu\paren{(a,b)\cap\brac{x\colon x>t}}e^tdt\\
  &=\int _a^b(b-t)^2e^tdt+\int _{-\infty}^a(b-a)^2e^tdt\\
  &=2\sqbrac{e^b-e^a-e^a(b-a)}
\end{align*}
In particular,
\begin{equation*}
\int _0^be^xdx=2(e^b-1-b)
\end{equation*}

For the Lebesgue integral we have the formula
\begin{equation*}
\int _a^bf(x)dx=\int _0^bf(x)dx-\int _0^af(x)dx
\end{equation*}
which is frequently used to simplify computations. This formula does not hold for our $q$-measure $\mu$. However, we do have the following result.

\begin{thm}       
\label{thm51}
If $f$ is increasing, differentiable, nonnegative on $\sqbrac{0,1}$ and $f^{-1}(\infty )\ge b$, $f^{-1}(0)\le a$, then
\begin{equation*}
\int _a^bfd\mu =\int _0^bfd\mu -\int _0^afd\mu -2(b-a)\int _0^af(t)dt
\end{equation*}
\end{thm}
\begin{proof}
Employing the change of variable formula gives
\begin{align*}
\int _a^bfd\mu&=\int _{f^{-1}(0)}^{f^{-1}(\infty )}\mu\paren{(a,b)\cap\brac{x\colon x>t}}f'(t)dt\\
  &=\int _a^b(b-t)^2f'(t)dt+\int _{f^{-1}(0)}^a(b-a)^2f'(t)dt\\
  &=\int _0^b(b-t)^2f'(t)dt-\int _0^a(b-t)^2f'(t)dt+(b-a)^2f(a)
\end{align*}
On the other hand, using integration by parts we have
\begin{align*}
\int _a^bfd\mu-\int _0^afd\mu&=\int _0^b(b-t)^2f'(t)dt+b^2f(0)-\!\int _0^a\!(a-t)^2f'(t)dt-a^2f(0)\\
  &=\int _0^b(b-t)^2f'(t)dt-\int _0^a(b-t)^2f'(t)dt\\
  &\quad +\int _0^a\sqbrac{(b-t)^2-(a-t)^2}f'(t)dt+(b^2-a^2)f(0)\\
  &=\int _0^b(b-t)^2f'(t)dt-\int _0^a(b-t)^2f'(t)dt\\
  &\quad +(b^2-a^2)f(a)-2(b-a)\int _0^atf'(t)dt\displaybreak\\
  &=\int _0^a(b-t)^2f'(t)dt-\int _0^a(b-t)^2f'(t)dt\\
  &\quad +(b-a)^2f(a)+2(b-a)\int _0^af(t)dt
\end{align*}
The result now follows.
\end{proof}

\begin{exam}{3}    
In this example we use some previous computations to verify Theorem~\ref{thm51}. We have shown that
\begin{equation*}
\int _a^bx^nd\mu =\frac{2}{(n+1)(n+2)}\,b^{n_2}
\end{equation*}
Hence, by Theorem~\ref{thm51} we have
\begin{align*}
\int _a^bx^nd\mu&=\int _0^bx^nd\mu -\int _0^ax^nd\mu -2(b-a)\int _0^at^ndt\\
  &=\frac{2}{(n+1)(n+2)}\,(b^{n+2}-a^{n+2})-\frac{2a^{n+1}}{n+1}\,(b-a)
\end{align*}
which agrees with our previous result. We have shown that
\begin{equation*}
\int _a^be^xd\mu =2(e^b-1-b)
\end{equation*}
Hence, by Theorem~\ref{thm51} we have
\begin{align*}
\int _a^be^xd\mu&=\int _0^be^xd\mu -\int _0^ae^xd\mu -2(b-a)\int _0^ae^tdt\\
  &=2(e^b-1-b)-2(e^a-1-a)-2(b-a)(e^a-1)\\
  &=2\sqbrac{e^b-e^a-e^a(b-a)}
\end{align*}
which agrees with our previous result.
\end{exam}

\begin{exam}{4}    
We compute the quantum integral of $f(x)=x+x^2$. By the change of variable formula we have
\begin{align*}
\int _0^b(x+x^2)d\mu&=\int _0^b(b-t)^2(1+2t)dt=\int _0^b(b-t)^2dt+2\int _0^b(b-t)^2tdt\\
  &=\frac{b^3}{3}+\frac{b^4}{6}
\end{align*}
This gives the surprising result that
\begin{equation*}
\int _0^b(x+x^2)d\mu =\int _0^bxd\mu +\int _0^bx^2d\mu
\end{equation*}
We shall later show that this quantum integral is always additive for sums of increasing continuous functions even if they are not differentiable. The next example shows that this result does not hold for two monomials if their sum is not increasing.
\end{exam}

\begin{exam}{5}    
Let $f(x)=x-x^2$ for $x\in\sqbrac{0,1}$. To evaluate $\int _0^bf(x)d\mu$ we cannot use the change of variable formula because $f$ is not increasing, so we will proceed directly. Let $1/2\le b\le 1$. For $0\le\lambda\le 1/4$ we have that
$\lambda =x-x^2$, if and only if $x=\paren{1\pm\sqrt{1-4\lambda\,}\,}/2$. Hence, for $\lambda\ge b-b^2$ we have
\begin{equation*}
\nu\paren{(0,b)\cap\brac{x\colon x-x^2>\lambda}}=
\begin{cases}
\sqrt{1-4\lambda\,}&\text{if }0\le\lambda\le 14\\0&\text{if }1/4\le\lambda\le 1
\end{cases}
\end{equation*}
and for $\lambda\le b-b^2$ we have
\begin{equation*}
\nu\paren{(0,b)\cap\brac{x\colon x=x^2>\lambda}}=b-\tfrac{1}{2}+\tfrac{1}{2}\,\sqrt{1-4\lambda\,}
\end{equation*}
Hence,
\begin{align*}
\int _0^b(x-x^2)d\mu&=\int _0^{b-b^2}\paren{b-\tfrac{1}{2}+\tfrac{1}{2}\,\sqrt{1-4\lambda\,}}^2d\lambda
  +\int _{b-b^2}^{1/4}(1-4\lambda )d\lambda\\
  &=-\frac{1}{24}+\frac{1}{3}\,b-b^2+\frac{5}{3}\,b^3-\frac{5}{6}\,b^4
\end{align*}
Notice that this does not coincide with
\begin{equation*}
\int _0^bxd\mu -\int _0^bx^2d\mu =\frac{1}{3}\,b^3-\frac{1}{6}\,b^4
\end{equation*}
For completeness we evaluate the integral with $0\le b\le 1/2$. Since $f$ is increasing on this interval we obtain the expected result:
\begin{equation*}
\int _0^b(x-x^2)d\mu =\int _0^{b-b^2}\paren{b-\tfrac{1}{2}+\tfrac{1}{2}\,\sqrt{1-4\lambda\,}\,}^2d\lambda
  =\tfrac{1}{3}\,b^3-\tfrac{1}{6}\,b^4
\end{equation*}
\end{exam}

\begin{exam}{6}    
Let $f$ be the following piecewise linear function:
\begin{equation*}
f(x)=\begin{cases}2x&\text{if }0\le x\le 1/2\\2-2x&\text{if }1/2\le x\le 1
\end{cases}
\end{equation*}
Let $1/2\le b\le 1$. For $0\le\lambda\le 2-2b$ we have
\begin{equation*}
\nu\paren{(0,b)\cap\brac{x\colon f(x)>\lambda}}=b-\tfrac{\lambda}{2}
\end{equation*}
and for $2-2b\le\lambda\le 1$ we have
\begin{equation*}
\nu\paren{(0,b)\cap\brac{x\colon f(x)>\lambda}}=1-\lambda
\end{equation*}
Hence
\begin{equation*}
\int _0^bfd\mu =\int _0^{2-2b}\paren{b-\tfrac{\lambda}{2}}^2d\lambda +\int _{2-2b}^1(1-\lambda )^2d\lambda
=\tfrac{1}{3}-2b+4b^2-2b^3
\end{equation*}
If $0\le b\le 1/2$ we obtain the expected result
\begin{equation*}
\int _0^bfd\mu =\int _0^{2b}\paren{b-\tfrac{\lambda}{2}\,}^2d\lambda =\tfrac{2}{3}\,b^3
\end{equation*}
\end{exam}

Observe that
\begin{align*}
\frac{1}{2}\,\frac{d^2}{db^2}\int _0^bx^nd\mu&=b^n\\
\frac{1}{2}\,\frac{d^2}{db^2}\int _0^be^xd\mu&=e^b\\
\end{align*}
However, in Example~5 we have for $b>1/2$ that
\begin{equation*}
\frac{1}{2}\,\frac{d^2}{db^2}\int _0^b(x-x^2)d\mu =-1+5b-5b^3\ne b-b^2
\end{equation*}
and in Example 6 we have for $b>1/2$ that
\begin{equation*}
\frac{1}{2}\,\frac{d^2}{db^2}\int _0^bfd\mu =4-6b\ne 2-2b=f(b)
\end{equation*}
These examples again illustrate the special nature of increasing functions. The next result is a quantum counterpart to the fundamental theorem of calculus.

\begin{thm}       
\label{thm52}
If $f$ is continuous and monotone on $\sqbrac{0,1}$, then
\begin{equation*}
\frac{1}{2}\,\frac{d^2}{db^2}\int _0^bfd\mu =f(b)
\end{equation*}
\end{thm}
\begin{proof}
If $f$ is decreasing then $-f$ is increasing so we can assume $f$ is increasing. For a positive integer $n$, let $g$ be the following increasing step function on $\sqbrac{0,1}$:
\begin{equation*}
g=c_1\chi _{\sqbrac{0,1/n}}+c_2\chi _{\parsq{1/n,2/n}}+\cdots +c_n\chi _{\parsq{(n-1)/n,1}}
\end{equation*}
where $0<c_1<\cdots  <c_n$. For $0<b\le 1$ we have that $(m-1)/n<b\le m/n$ for some integer $0<m\le n$ and
\begin{equation*}
g\chi _{\sqbrac{0,b}}=c_1\chi _{\sqbrac{0,1/n}}+c_2\chi _{\parsq{1/n,2/n}}
  +\cdots +c_{m-1}\chi _{\parsq{(m-2)/n,(m-1)/n}}+c_m\chi _{\parsq{(m-1)/n,b}}
\end{equation*}
Letting $A_i=\parsq{(i-1)/n,i/n}$, $i=1,\ldots ,m-1$, $A_m=\parsq{(m-1)/n,b}$ and $\bhat =b-(m-1)/n$ we have by \eqref{eq24} of Lemma~\ref{lem21} that
\begin{align*}
\int _0^bgd\mu&=c_1\left [\mu\paren{A_1\cupdot A_2}+\cdots +\mu\paren{A_1\cupdot A_m}\right.\\
  &\quad\left. -(m-2)\mu (A_1)-\mu (A_2)-\cdots -\mu (A_m)\right ]\\
  &\quad +c_2\left [\mu\paren{A_2\cupdot A_3}+\cdots +\mu\paren{A_2\cupdot A_m}\right.\\
  &\quad\left. -(m-3)\mu (A_2)-\mu (A_3)-\cdots -\mu (A_m)\right ]\\
  &\quad +\cdots +c_{m-1}\sqbrac{\mu\paren{A_{m-1}\cupdot A_m}-\mu (A_m)}+c_m\mu (A_m)\\
  &=c_1\sqbrac{(m-2)\paren{\frac{2}{n}}^2+\paren{\frac{1}{n}+\bhat}^2-(2m-4)\paren{\frac{1}{n}}^2-{\bhat\,}^2}\\
  &\quad +c_2\sqbrac{(m-3)\paren{\frac{2}{n}}^2+\paren{\frac{1}{n}+\bhat}^2-(2m-6)\paren{\frac{1}{n}}^2-{\bhat\,}^2}\\
  &\quad +\cdots +c_{m-1}\sqbrac{\paren{\frac{1}{n}+\bhat}^2-{\bhat\,}^2}+c_m{\bhat \,}^2\\
  &=c_1\sqbrac{(2m-3)\frac{1}{n^2}+\frac{2}{n}\,\bhat}+c_2\sqbrac{(2m-5)\frac{1}{n^2}+\frac{2}{n}\,\bhat}\\
  &\quad +\cdots +c_{m-1}\paren{\frac{1}{n^2}+\frac{2}{n}\,\bhat}+c_m{\bhat\,}^2
\end{align*}
It follows that
\begin{equation}
\label{eq51}                 
\frac{d}{db}\int _0^bgd\mu =\frac{2}{n}(c_1+c_2+\cdots +c_{m-1})+2c_m\bhat
\end{equation}
and that
\begin{equation}
\label{eq52}                 
\frac{1}{2}\frac{d^2}{db^2}\int _0^bgd\mu =c_m=g(b)
\end{equation}
We can assume without loss of generality that $f$ is nonnegative. Then there exists an increasing sequence of increasing nonnegative step functions $s_i$ converging uniformly to $f$. Since
\begin{equation*}
\mu\sqbrac{s_{i+1}^{-1}(\lambda ,\infty )}\ge\mu\sqbrac{s_i^{-1}(\lambda ,\infty )}
\end{equation*}
we have by the continuity of $\mu$ that
\begin{equation*}
\mu\sqbrac{f^{-1}(\lambda ,\infty )}=\mu\sqbrac{\cup s_i^{-1}(\lambda ,\infty )}
  =\lim\mu\sqbrac{s_i^{-1}(\lambda ,\infty )}
\end{equation*}
These same formulas apply to $f\chi _{\sqbrac{0,b}}$ and $s_i\chi _{\sqbrac{0,b}}$. By the quantum bounded monotone convergence theorem \cite{gud3} we conclude that
\begin{equation*}
\int _0^bfd\mu _i=\lim\int _0^bs_id\mu
\end{equation*}
Applying \eqref{eq51} with $g$ replaced by $s_i$ it can be checked that the sequence of functions of $b$ given by
\begin{equation*}
\frac{d}{db}\int _0^bs_id\mu
\end{equation*}
is uniformly Cauchy so the sequence converges and hence
\begin{equation*}
\frac{d}{db}\int _0^bfd\mu =\lim\frac{d}{db}\int _0^bs_id\mu
\end{equation*}
By \eqref{eq52}
\begin{equation*}
\frac{d^2}{db^2}\int _0^bs_id\mu
\end{equation*}
converges uniformly so we have
\begin{equation*}
\frac{1}{2}\frac{d^2}{db^2}\int _0^bfd\mu =\lim\frac{1}{2}\frac{d^2}{db^2}\int _0^bs_id\mu =f(b)\qedhere
\end{equation*}
\end{proof}

\begin{lem}       
\label{lem53}
If $f$ is continuous and monotone on $\sqbrac{0,1}$, then
\begin{equation*}
\sqbrac{\frac{d}{db}\int _0^bfd\mu}(0)=0
\end{equation*}
\end{lem}
\begin{proof}
We can assume without loss of generality that $f$ is increasing. Let $g=\sum c_i\chi _{A_i}$ be a step function as in the proof of Theorem~\ref{thm52}. If $b$ is sufficiently small we have
\begin{equation*}
g\chi _{\sqbrac{0,b}}=c_1\chi _{A_1\cap\sqbrac{0,b}}=c_1\chi _{\sqbrac{0,b}}
\end{equation*}
Hence, for such $b$ we have
\begin{equation*}
\int _0^bgd\mu =\int g\chi _{\sqbrac{0,b}}d\mu =c_1\int\chi _{\sqbrac{0,b}}d\mu =c_1b^2
\end{equation*}
Therefore,
\begin{equation*}
\sqbrac{\frac{d}{db}\int _0^bgd\mu}(0)=\paren{\frac{d}{db}c_1b^2}(0)=0
\end{equation*}
As shown in the proof of Theorem~\ref{thm52}, there exists an increasing sequence of step functions $s_i$ such that
\begin{equation*}
\frac{d}{db}\int _0^bfd\mu =\lim\frac{d}{db}\int _0^bs_id\mu
\end{equation*}
The result follows.
\end{proof}

Part (b) of the next theorem is the second half of the quantum fundamental theorem of calculus.

\begin{thm}       
\label{thm54}
{\rm (a)}\enspace If $f$ is continuous and monotone on $\sqbrac{0,1}$, then
\begin{equation*}
\int _0^bfd\mu =2\int _0^b\int _0^tf(x)dxdt
\end{equation*}
{\rm (b)}\enspace If $f''$ is monotone and continuous on $\sqbrac{0,1}$, then
\begin{equation*}
\int _0^b\frac{1}{2}f''d\mu =f(b)-f(0)-f'(0)b
\end{equation*}
\end{thm}
\begin{proof}
(a)\enspace If $g''=f$, then integrating gives
\begin{equation*}
\int _0^tf(x)dx=g'(t)-g'(0)
\end{equation*}
Integrating again we have
\begin{equation*}
\int _0^b\int _0^tf(x)dxdt=g(b)-g(0)-g'(0)b
\end{equation*}
Hence, for all $b\in\sqbrac{0,1}$ we have
\begin{equation*}
g(b)=g(0)+g'(0)b+\int _0^b\int _0^tf(x)dxdt
\end{equation*}
Since by Theorem~\ref{thm52}
\begin{equation*}
\frac{d^2}{db^2}\int _0^b\frac{1}{2}fd\mu =f(b)
\end{equation*}
letting $g(b)=\int _0^b\tfrac{1}{2}fd\mu$ we have that $g(0)=0$ and by Lemma~\ref{lem53} we obtain $g'(0)=0$. Hence,
\begin{equation*}
\int _0^b\int _0^tf(x)dxdt=g(b)=\frac{1}{2}\int _0^bfd\mu
\end{equation*}
(b)\enspace By Part (a) we have
\begin{align*}
\int _0^b\frac{1}{2}f''d\mu&=\int _0^b\int _0^tf''(x)dxdt=\int _0^b\sqbrac{f'(t)-f'(0)}dt\\
  &=f(b)-f(0)-f'(0)b\qedhere
\end{align*}
\end{proof}

The next corollary follows from Theorem~\ref{thm54}(a)

\begin{cor}       
\label{cor55}
The quantum (Lebesgue)${}^2$ integral is additive for increasing (decreasing) continuous functions.
\end{cor}

\begin{exam}{7}    
We compute some quantum integrals using Theorem~\ref{thm54}(a).
\begin{align*}
\int _0^b\cos xd\mu&=2\int _0^b\int _0^t\cos xdxdt=2\int _0^b\sin tdt=2(1-\cos b)\\
\int _0^b\sin xd\mu&=2\int _0^b\int _0^t\sin xdxdt=2\int _0^b(1-cos t)dt=2(b-\sin b)\\
\int _0^b\cosh\sqrt{2\,}\,xd\mu&=2\int _0^b\int _0^t\cosh\sqrt{2\,}\,xdxdt=\sqrt{2\,}\int _0^b\sinh\sqrt{2\,}\,tdt\\
  &=\cosh\sqrt{2\,}\,b-1
\end{align*}
The last integral shows that the quantum counterpart of $e^x$ is $\cosh\sqrt{2\,}\,x$.
\end{exam}
\eject
\noindent\textbf{\large Acknowledgement.}\enspace
The author thanks Petr Vojt\v echovsk\'y for pointing out reference \cite{hir96}

\end{document}